\documentclass[12pt,fulpage]{article}

\usepackage[hypertex,colorlinks]{hyperref}
\usepackage{fullpage,amssymb,amsmath,amsthm}
\usepackage{latexsym}
\usepackage{verbatim}

\newcommand{\remove}[1]{}

\newtheorem{theorem}{Theorem}[section]
\newtheorem{thm}{Theorem}[section]

\newtheorem{lemma}[thm]{Lemma}

\newtheorem{definition}[thm]{Definition}

\newtheorem{corollary}[thm]{Corollary}

\newtheorem{remark}[thm]{Remark}

\def\F{{\mathbb{F}}}

\def\poly{{\mathrm{poly}}}
\def\rank{{\mathrm{rank}}}

\def\B{{\{0,1\}}}

\begin{document}

\title{New constructions of WOM codes using the Wozencraft ensemble}

\author{Amir Shpilka\thanks{Faculty of Computer Science, Technion --- Israel
Institute of Technology, Haifa, Israel, {\tt
shpilka@cs.technion.ac.il}. This research was partially supported
by the Israel Science Foundation (grant number 339/10).}}


\date{}

\maketitle

\begin{abstract}

In this paper we give several new constructions of WOM codes. The
novelty in our constructions is the use of the so called {\em
Wozencraft ensemble} of linear codes. Specifically, we obtain the
following results.

We give an explicit construction of a two-write Write-Once-Memory
(WOM for short) code that approaches capacity, over the binary
alphabet. More formally, for every $\epsilon>0$, $0<p<1$ and $n
=(1/\epsilon)^{O(1/p\epsilon)}$ we give a construction of a
two-write WOM code of length $n$ and capacity $H(p)+1-p-\epsilon$.
Since the capacity of a two-write WOM code is $\max_p (H(p)+1-p)$,
we get a code that is $\epsilon$-close to capacity. Furthermore,
encoding and decoding can be done in time $O(n^2\cdot\poly(\log
n))$ and time $O(n\cdot\poly(\log n))$, respectively, and in
logarithmic space.

We obtain a new encoding scheme for $3$-write WOM codes over the
binary alphabet. Our scheme achieves rate $1.809-\epsilon$, when
the block length is $\exp(1/\epsilon)$. This gives a better rate
than what could be achieved using previous techniques.

We highlight a connection to linear seeded extractors for
bit-fixing sources. In particular we show that obtaining such an
extractor with seed length $O(\log n)$ can lead to improved
parameters for $2$-write  WOM codes. We then give an application
of existing constructions of extractors to the problem of
designing encoding schemes for memory with defects.

\end{abstract}

\section{Introduction}

In \cite{RivestShamir82} Rivest and Shamir introduced the notion
of {\em write-once-memory} and showed its relevance to the problem
of saving data on optical disks. A write-once-memory, over the
binary alphabet, allows us to change the value of a memory cell
(say from $0$ to $1$) only once. Thus, if we wish to use the
storage device for storing $t$ messages in $t$ rounds, then we
need to come up with an encoding scheme that allows for $t$-write
such that each memory cell is written at most one time. An
encoding scheme satisfying these properties is called a
Write-Once-Memory code, or a WOM code for short. This model has
recently gained renewed attention due to similar problems that
arise when using flash memory devices. We refer the readers to
\cite{YKSVW10} for a more detailed introduction to WOM codes and
their use in encoding schemes for flash memory.

One interesting goal concerning WOM codes is to find codes that
have good rate for $t$-write. Namely, to find encoding schemes
that allow to save the maximal information-theoretic amount of
data possible under the write-once restriction. Following
\cite{RivestShamir82} it was shown that the capacity (i.e. maximal
rate) of $t$-write  binary WOM code is\footnote{All logarithms in
this paper are taken base $2$.} $\log(t+1)$ (see
\cite{RivestShamir82,Heegard84,FuVinck99}). Stated differently, if
we wish to use an $n$-cell memory $t$-times then each time we can
store, on average, $n\cdot \log(t+1)/t$ many bits.

\remove{
Currently, the
best known explicit encoding scheme for two-write  (over the
binary alphabet) has rate roughly $1.49$ (compared to the optimal
$\log 3 \approx 1.585$) \cite{YKSVW10}. We note that these codes,
of rate $1.49$, were found using the help of a computer search. A
more `explicit' construction given in  \cite{YKSVW10} achieves
rate $1.46$.
}

In this work we address the problem of designing WOM codes that
achieve the theoretical capacity for the case of two rounds of
writing to the memory cells. Before describing our results we give
a formal definition of a two-write WOM code.

For two vectors of the same length $y$ and $y'$ we say that
$y'\leq y$ if $y'_i\leq y_i$ for every coordinate $i$.


\begin{definition}
A two-write binary WOM of length $n$ over the sets of messages
$\Omega_1$ and $\Omega_2$ consists of two encoding functions
$E_1:\Omega_1 \to \B^n$ and $E_2: E_1(\Omega_1)\times \Omega_2 \to
\B^n$ and two decoding functions $D_1:E_1(\Omega_1) \to \Omega_1$
and $D_2: E_2(E_1(\Omega_1)\times \Omega_2) \to \Omega_2$ that
satisfy the following properties.
\begin{enumerate}
  \item For every $x\in \Omega_1$, $D_1(E_1(x))=x$.
  \item For every $x_1\in \Omega_1$ and $x_2\in \Omega_2$, we have
  that $E_1(x_1)\leq E_2(E_1(x_1),x_2)$.
  \item For every $x_1\in \Omega_1$ and $x_2\in \Omega_2$, it
  holds that $D_2(E_2(E_1(x_1),x_2))=x_2$.
\end{enumerate}
The rate of such a WOM code is defined to be $(\log|\Omega_1| +
\log|\Omega_2|)/n$.
\end{definition}

Intuitively, the definition enables the encoder to use $E_1$ as
the encoding function in the first round. If the message $x_1$ was
encoded (as the string $E_1(x_1)$) and then we wished to encode in
the second round the message $x_2$, then we write the string
$E_2(E_1(x_1),x_2)$. Since $E_1(x_1)\leq E_2(E_1(x_1),x_2)$, we
only have to change a few zeros to ones in order to move from
$E_1(x_1)$ to $E_2(E_1(x_1),x_2)$. The requirement on the decoding
functions $D_1$ and $D_2$ guarantees that at each round we can
correctly decode the memory.\footnote{We implicitly assume that
the decoder knows, given a codeword, whether it was encoded in the
first or in the second round. At worst this can add another bit to
the encoding and has no affect (in the asymptotic sense) on the
rate.} Notice that in the second round we are only required to
decode $x_2$ and not the pair $(x_1,x_2)$. It is not hard to see
that insisting on decoding both $x_1$ and $x_2$ is a too strong
requirement that does not allow rate more than $1$.

The definition of a $t$-write code is similar and is left to the
reader. Similarly, one can also define WOM codes over other
alphabets, but in this paper we will only be interested in the
binary alphabet.

In \cite{RivestShamir82} it was shown that the maximal rate (i.e.
the capacity) that a WOM code can have is at most $\max_p
H(p)+(1-P)$ where $H(p)$ is the entropy function. It is not hard
to prove that this expression is maximized for $p=1/3$ and is
equal to $\log 3$. Currently, the best known explicit encoding
scheme for two-write (over the binary alphabet) has rate roughly
$1.49$ (compared to the optimal $\log 3 \approx 1.585$)
\cite{YKSVW10}. We note that these codes, of rate $1.49$, were
found using the help of a computer search. A more `explicit'
construction given in \cite{YKSVW10} achieves rate $1.46$.

Rivest and Shamir were also interested in the case where both
rounds encode the same amount of information. That is,
$|\Omega_1|=|\Omega_2|$. They showed that the rate of such codes
is at most $H(p)+1-p$, for $p$ such that $H(p)=1-p$ ($p\approx
0.227$). Namely, the maximal possible rate is roughly $1.5458$.
Yaakobi et al. described a construction (with
$|\Omega_1|=|\Omega_2|$) that has rate $1.375$ and mentioned that
using a computer search they found such a construction with rate
$1.45$ \cite{YKSVW10}.

\subsection{Our results}\label{sec: our results}

Our main theorem concerning $2$-write WOM codes over the binary alphabet
is the following.

\begin{theorem}\label{thm:main}
  For any $\epsilon>0$, $0<p<1$ and $c>0$ there is $N=N(\epsilon,p,c)$ such that for
  every $n>N(\epsilon,p,c)$ there is an explicit construction of a
  two-write  WOM code of length $n(1+o(1))$ of rate at least $H(p)+1-p -\epsilon$.
  Furthermore, the encoding function can be computed in time $ n^{c+1}\cdot \poly(c\log n)$
  and decoding can be done in time $n \cdot \poly(c\log n)$. Both
  encoding and decoding can be done in logarithmic space.
\end{theorem}

In particular, for $p=1/3$ we give a construction of a WOM code
whose rate is $\epsilon$ close to the capacity. If we wish to
achieve a polynomial time encoding and decoding then our proof
gives the bound $N(\epsilon,p,c)=
(c\epsilon)^{-O(1/(c\epsilon))}$. If we wish to have a short block
length, i.e. $n=\poly(1/\epsilon)$, then our running time
deteriorates and becomes $n^{O(1/\epsilon)}$.

In addition to giving a new approach for constructing
capacity approaching WOM codes we also demonstrate a
method to obtain capacity approaching codes from existing
constructions (specifically, using the methods of  \cite{YKSVW10})
without storing huge lookup tables. We explain this scheme
in Section~\ref{sec:no lookup}.
\\

Using our techniques we obtain the following result
for $3$-write WOM codes over the binary alphabet.

\begin{theorem}\label{thm:3-write}
  For any $\epsilon>0$, there is $N=N(\epsilon)$ such that for
  every $n>N(\epsilon,p,c)$ there is an explicit construction of a
  $3$-write WOM code of length $n$ that has rate larger than $1.809-\epsilon$.
\end{theorem}

Previously the best construction of $3$-write WOM codes over the
binary alphabet had rate $1.61$ \cite{KYSVW10}. Furthermore, the
technique of \cite{KYSVW10} cannot provably yield codes that have
rate larger than $1.661$. Hence, our construction yields a higher
rate than the best possible rate achievable by previous methods.
However, we recall that the capacity of $3$-write WOM codes over
the binary alphabet is $\log (3+1) =2$. Thus, even using our new
techniques we fall short of achieving the capacity for this case.
The proof of this result is given in Section~\ref{sec:3-write}.\\

In addition to the results above, we highlight a connection
between schemes for $2$-write WOM codes and {\em extractors for
bit-fixing sources}, a combinatorial object that was studied in
complexity theory (see Section~\ref{sec:extractors} for
definitions). We then use this connection to obtain new schemes
for dealing with defective memory. This result is described in
Section~\ref{sec:defective} (see Theorem~\ref{thm:defective}).

\remove{
Rivest and Shamir were also interested in the case where both
rounds encode the same amount of information. That is,
$|\Omega_1|=|\Omega_2|$. They showed that the rate of such codes
is $H(p)+1-p$, for $p$ such that $H(p)=1-p$ ($p\approx 0.227$).
Namely, the rate is roughly $1.5458$. Thus, our construction gives
rise to a code that is $\epsilon$-close to capacity also in this
scenario.
Yaakobi et al. described a construction (with
$|\Omega_1|=|\Omega_2|$) that has rate $1.375$ and mentioned that
using a computer search they found such a construction with rate
$1.45$ \cite{YKSVW10}. To the best of our knowledge this is the
best existing construction.
}

\subsection{Is the problem interesting?}\label{sec:discussion}

The first observation that one makes is that the problem of
approaching capacity is, in some sense, trivial. This basically
follows from the fact that concatenating WOM codes (in the sense
of string concatenation) does not hurt any of their properties.
Thus, if we can find, even in a brute force manner, a code of
length $m$ that is $\epsilon$-close to capacity, in time $T(m)$,
then concatenating $n=T(m)$ copies of this code, gives a code of
length $nm$ whose encoding algorithm takes $nT(m)=n^2$ time.
Notice however, that for the brute force algorithm, $T(m) \approx
2^{2^m}$ and so, to get $\epsilon$-close to capacity we need
$m\approx 1/\epsilon$ and thus $n \approx 2^{2^{1/\epsilon}}$.

The same argument also shows that finding capacity approaching WOM
codes for $t$-write, for any constant $t$, is ``easy'' to achieve
in the asymptotic sense, with a polynomial time encoding/decoding
functions, given that one is willing to let the encoding length
$n$ be obscenely huge.

In fact, following Rivest and Shamir,
Heegard actually showed that a randomized encoding scheme can
achieve capacity for all $t$ \cite{Heegard84}.

In view of that, our construction can be seen as giving a big
improvement over the brute force construction. Indeed, we only
require $n \approx {2^{1/\epsilon}}$ and we give encoding and
decoding schemes that can be implemented in logarithmic space.
Furthermore, our construction is highly {\em structured}. This
structure perhaps could be used to find ``real-world'' codes with
applicable parameters. Even if not, the ideas that are used in our
construction can be helpful in designing better WOM codes of
reasonable lengths.

We later discuss a connection with {\em linear seeded extractors
for bit-fixing sources}. A small improvement to existing
constructions could lead to capacity-achieving WOM codes of
reasonable block length.

\subsection{Organization}

We start by describing the method of \cite{CohenGodlewskiMekx,Wu10,YKSVW10} in
Section~\ref{sec:Yaakobi} as it uses similar ideas to our
construction. We then give an overview of our construction in
Section~\ref{sec:method} and the actual construction and its
analysis in Section~\ref{sec:construction}. In
Section~\ref{sec:extractors} we discuss the connection to
extractors and then show the applicability of extractors
for dealing with defective memories in Section~\ref{sec:defective}.
In Section~\ref{sec:no lookup} we show how one
can use the basic approach of \cite{YKSVW10} to achieve
capacity approaching WOM codes that do not need
large lookup tables.
Finally, we prove Theorem~\ref{thm:3-write} in Section~\ref{sec:3-write}.

\subsection{Notation}
For a $k\times m$ matrix $A$ and a subset $S\subset [m]$ we let
$A|_S$ be the $k\times |S|$ submatrix of $A$ that contains only
the columns that appear in $S$. For a length $m$ vector $y$ and a
subset $S \subset [m]$ we denote with $y|_S$ the vector that is
equal to $y$ on all the coordinates in $S$ and that has zeros
outside $S$.

\section{The construction of \cite{CohenGodlewskiMekx,Wu10,YKSVW10}}\label{sec:Yaakobi}

As it turns out, our construction is related to the construction
of WOM codes of Cohen et al. \cite{CohenGodlewskiMekx} as well as
to that of Wu \cite{Wu10} and of Yaakobi et al.
\cite{YKSVW10}.\footnote{Cohen et al. first did it for $t>2$ and
then Wu used it for $t=2$. Wu's ideas were then slightly refined
by Yaakobi et al.} We describe the idea behind the construction of
 Yaakobi et al. next (the constructions of \cite{CohenGodlewskiMekx,Wu10} are similar). Let
$0<p<1$ be some fixed number.

Similarly to \cite{RivestShamir82}, in the first round
\cite{YKSVW10} think of a message as a subset $S\subset [n]$ of
size $pn$ and encode it by its characteristic vector.
Clearly in this step we can transmit $H(p)n$ bits of information.
(I.e. $\log|\Omega_1| \approx H(p)n$.)

For the second round assume that we already send a message
$S\subset [n]$. I.e. we have already written $pn$ locations. Note
that in order to match the capacity we should find a way to
optimally use the remaining $(1-p)n$ locations in order to
transmit $(1-p-o(1))n$ many bits.
Imagine that we have a binary MDS code. Such codes of course do
not exist but for the sake of explanations it will be useful to
assume their existence. Recall that a linear MDS code of rate
$n-k$ can be described by a $k\times n$ parity check matrix $A$
having the property that any $k$ columns have full rank. I.e. any
$k\times k$ submatrix of $A$ has full rank. Such matrices exist
over large fields (i.e. parity check matrices of Reed-Solomon
codes) but they do not exist over small fields. Nevertheless,
assume that we have such a matrix $A$ that has $(1-p)n$ rows.
Further, assume that in the first round we transmitted a word
$w\in\{0,1\}^n$ of weight $|w|=pn$ representing a set $S$. Given a
message $x\in\{0,1\}^{(1-p)n}$ we find the {\em unique} $y\in
\{0,1\}^n$ such that $Ay=x$ and $y|_S=w$. Notice that the fact
that each $(1-p)n\times (1-p)n$ submatrix of $A$ has full rank
guarantees the existence of such a $y$. Our encoding of $x$ will
be the vector $y$. When the decoder receives a message $y$ in
order to recover $x$ she simply computes $Ay$. As we did not touch
the nonzero coordinates of $w$ this is a WOM encoding scheme.

As such matrices $A$ do not exist, Yaakobi et al. look for
matrices that have many submatrices of size $(1-p)n\times (1-p)n$
that are full rank and restrict their attention only to sets $S$
such that the set of columns corresponding to the complement of
$S$ has full rank. (I.e. they modify the first round of
transmission.) In principal, this makes the encoding of the first
round highly non-efficient as one needs a lookup table in order to
store the encoding scheme. However, \cite{YKSVW10} showed that
such a construction has the ability to approach capacity. For
example, if the matrix $A$ is randomly chosen among all
$(1-p)n\times n$ binary matrices then the number of $(1-p)n\times
(1-p)n$ submatrices of $A$ that have full rank is roughly
$2^{H(p)n}$.

\begin{remark}
Similar to the concerns raised in Section~\ref{sec:discussion},
this method (i.e. picking a random matrix, verifying that it has
the required properties and encoding the ``good'' sets of columns)
requires high running time in order to get codes that are
$\epsilon$-close to capacity. In particular, one has to go over
all matrices of dimension, roughly, $1/\epsilon \times
O(1/\epsilon)$ in order to find a good matrix which takes time
$\exp(1/\epsilon^2)$. Furthermore, the encoding scheme requires a
lookup table whose space complexity is $\exp(1/\epsilon)$. Thus,
even if we use the observation raised in
Section~\ref{sec:discussion} and concatenate several copies of
this construction in order to reach a polynomial time encoding
scheme, it will still require a large space. (And the block length
will even be slightly larger than in our construction.)

Nevertheless, in Section~\ref{sec:no lookup} we show
how one can trade space for computation. In other words,
we show how one can approach capacity using this approach
without the need to store huge lookup tables.
\end{remark}

\section{Our method}\label{sec:method}

We describe our technique for proving Theorem~\ref{thm:main}.
The main idea is that we can use a collection of binary codes that
are, in some sense, {\em MDS codes on average}. Namely, we show a
collection of (less than) $2^{m}$ matrices $\{A_i\}$ of size
$(1-p-\epsilon)m\times m$ such that for any subset $S\subset[m]$,
of size $pm$, all but a fraction $2^{-\epsilon m}$ of the matrices
$A_i$, satisfy that $A_i|_{[m]\setminus S}$ has full row rank
(i.e. rank $(1-p-\epsilon)m$). Now, assume that in the first round
we transmitted a word $w$ corresponding to a subset $S\subset[m]$
of size $pm$. In the second round we find a matrix $A_i$ such that
$A_i|_{[m]\setminus S}$ has full row rank. We then use the same
encoding scheme as before. However, as the receiver does not know
which matrix we used for the encoding, we also send the ``name''
of the matrix alongside our message (using additional $m$ bits).

This idea has several drawbacks. First, to find the good matrix we
have to check $\exp(m)$ many matrices which takes a long time.
Secondly, sending the name of the matrix that we use require
additional $m$ bits which makes the construction very far from
achieving capacity.

To overcome both issues we note that we can in fact use the same
matrix for many different words $w$. However, instead of
restricting our attention to only one matrix and the sets of $w$'s
that is good for it, as was done in \cite{YKSVW10}, we change the encoding
in the following way. Let $M=m\cdot 2^{\epsilon m}$. In the first
step we think of each message as a collection of $M/m$ subsets
$S_1,\ldots,S_{M/m} \subset [m]$, each of size $pm$. Again we
represent each $S_i$ using a length $m$ binary vector of weight
$pm$, $w_i$. We now let $w=w_1 \circ w_2 \circ \ldots \circ
w_{M/m}$, where $a\circ b$ stands for string concatenation. For
the second stage of the construction we find, for a given
transmitted word $w\in \{0,1\}^M$, a matrix $A$ from our
collection such that all the matrices $A_{S_i}$ have full rank.
Since, for each set $S$ only $2^{-\epsilon m}$ of the matrices are
``bad'', we are guaranteed, by the union bound, that such a good
matrix exists in our collection. Notice that finding the matrix
requires time $\poly(M,2^m) = M^{O(1/\epsilon)}$. Now, given a
length $(1-p-\epsilon)M$ string $x=x_1\circ \ldots\circ x_{M/m}$
represented as the concatenation of $M/m$ strings of length
$(1-p-\epsilon)m$ each, we find for each $w_i$ a word $y_i\in
\{0,1\}^m$ such that $Ay_i=x_i$ and $y_i|_{S_i}=w_i$. Our encoding
of $x$ is $y_1\circ\ldots\circ y_{M/m}\circ I(A)$ where by $I(A)$
we mean the length $m$ string that serves as the index of $A$.
Observe that this time sending the index of $A$ has almost no
effect on the rate (the encoding length is $M = \exp(m)$ and the
``name'' of $A$ consists of at most $m$ bits). Furthermore, the
number of messages that we encode in the first round is equal to
${m\choose pm}^{M/m} = 2^{(H(p)-o(1))m\cdot
M/m}=2^{(H(p)-o(1))M}$. In the second round we clearly send an
additional $(1-p-\epsilon)M$ bits and so we achieve rate
$H(p)+(1-p-\epsilon)-o(1)$ as required.

However, there is still one drawback which is the fact that the
encoding requires $M^{1/\epsilon}$ time. To handle this we note
that we can simply concatenate $M^{1/\epsilon}$ copies of this
basic construction to get a construction of length
$n=M^{1+1/\epsilon}$ having the same rate, such that now encoding
requires time $M^{O(1/\epsilon)}=\poly(n)$.

We later use a similar approach, in combination with the Rivest-Shamir
encoding scheme, to prove Theorem~\ref{thm:3-write}.

\section{Capacity achieving $2$-write WOM codes}\label{sec:construction}

\subsection{Wozencraft ensemble}

We first discuss the construction known as Wozencraft's ensemble.
This will constitute our set of ``average'' binary MDS codes.

The Wozencraft ensemble consists of a set of $2^n$ binary codes of
block length $2n$ and rate $1/2$ (i.e. dimension $n$) such that
most codes in the family meet the Gilbert-Varshamov bound. To the
best of our knowledge, the construction known as Wozencraft's
ensemble first appeared in a paper by Massey \cite{Massey63}. It
later appeared in a paper of Justesen \cite{Justesen72} that
showed how to construct codes that achieve the Zyablov bound
\cite{Zyablov71}.

Let $k$ be a positive integer and $\F=\F_{2^k}$ be the field with
$2^k$ elements. We fix some canonical invertible linear map
$\sigma_k$ between $\F$ and $\F_2^k$ and from this point on we
think of each element $x\in\F$ both as a field element and as a
binary vector of length $k$, which we denote $\sigma_k(x)$. Let
$b>0$ be an integer. Denote $\pi_b:\{0,1\}^*\to \{0,1\}^b$ be the
map that projects each binary sequence on its first $b$
coordinates.

For two integers $0<b\leq k$, the $(k,k+b)$-Wozencraft ensemble is
the following collection of $2^k$ matrices. For $\alpha\in\F$
denote by $A_{\alpha}$ the unique matrix satisfying
$\sigma_k(x)\cdot A_{\alpha} = (\sigma_k(x), \pi_b
(\sigma_k(\alpha x)))$ for every $x\in\F$.

The following lemma is well known. For completeness we provide the
proof below.

\begin{lemma}\label{lem:woz-property}
For any $0\neq y\in\B^{k+b}$ the number of matrices $A_{\alpha}$
that $y$ is contained in the span of their rows is exactly
$2^{k-b}$.
\end{lemma}

\begin{proof}
Let us first consider the case where $b=k$, i.e., that we keep all
of $\sigma_k(\alpha x)$. In this case $\sigma_k(x)\cdot A_{\alpha}
= (\sigma_k(x),\sigma_k(\alpha x))$. Given $\alpha\neq\beta$ and
$x,y\in\B^{k}$ notice that if $\sigma_k(x)\cdot
A_{\alpha}=\sigma_k(y)\cdot A_{\alpha}$ then it must be the case
that $\sigma_k(x)=\sigma_k(y)$ and hence $x=y$. Now, if $x=y$ and
$0\neq x$ then since $\alpha\neq\beta$ we have that $\alpha x \neq
\beta x =\beta y$. It follows that the only common vector in the
span of the rows of $A_\alpha$ and $A_\beta$ is the zero vector
(corresponding to the case $x=0$).

Now, let use assume that $b\leq k$. Fix some $\alpha\in\F$ and let
$(\sigma_k(x), \pi_b(\sigma_k(\alpha x)))$ be some nonzero vector
spanned by the rows of $A_\alpha$. For any vector $u\in\B^{k-b}$
let $\beta_u \in \F$ be the {\em unique} element satisfying
$\sigma_k(\beta_u x) = \pi_b(\sigma_k(\alpha x))\circ u$. Notice
that such a $\beta_u$ exists and equal to $\beta_u =
\sigma^{(-1)}(\pi_b(\sigma_k(\alpha x))\circ u)\cdot x^{-1}$
($x\neq 0$ as we started from a nonzero vector in the row space of
$A_\alpha$). We thus have that $\sigma_k(x)\cdot A_{\beta_u} =
(\sigma_k(x), \pi_b (\sigma_k(\beta_u x))) = (\sigma_k(x),
\pi_b(\sigma_k(\alpha x)))$. Hence, $(\sigma_k(x),
\pi_b(\sigma_k(\alpha x)))$ is also contained in the row space of
$A_{\beta_u}$. Since this was true for any $u\in\B^k$, and clearly
for $u\neq u'$, $\beta_u \neq \beta_{u'}$ we see that any such row
is contained in the row space of exactly $2^{k-b}$ matrices
$A_\beta$.

\sloppy It is now also clear that there is no additional matrix
that contains $(\sigma_k(x), \pi_b(\sigma_k(\alpha x)))$ in its
row space. Indeed, if $A_\gamma$ is a matrix containing the vector
in its row space, then let $u$ be the last $k-b$ bits of
$\sigma_k(\gamma x)$. It now follows that $\sigma_k(\gamma x)=
\sigma_k(\beta_u x)$ and since $\sigma_k$ is an invertible linear
map and $x\neq 0$ this implies that $\gamma=\beta_u$.
\end{proof}

\begin{corollary}\label{cor:wozen-union bound}
  Let $y\in\B^{k+b}$ have weight $s$. Then, the number of matrices
  in the $(k,k+b)$-Wozencraft ensemble
  that contain a vector $0\neq y'\leq y$ in the span of their rows
  is at most $(2^s-1) \cdot 2^{k-b}< 2^{k+s-b}$.
\end{corollary}

To see why this corollary is relevant we prove the following easy
lemma.

\begin{lemma}\label{lem:good subset}
Let $A$ be a $k\times (k+b)$ matrix of full row rank ( i.e.
$\rank(A)=k$) and $S\subset[k+b]$  a set of columns. Then $A_S$
has full row rank if and only if there is no vector $y\neq0$
supported on $[k+b]\setminus S$ that is in the span of the rows of
$A$.
\end{lemma}

\begin{proof}
Assume that there is a nonzero vector $y$ in the row space of $A$
that is supported on $[k+b]\setminus S$. Hence, it must be the
case that $xA_S=0$. Since $x\neq 0$, this means that the rows of
$A_S$ are linearly dependent and hence $A_S$ does not have full
row rank.

To prove the other direction notice that if $\rank(A_S)<k$ then
there must be a nonzero $x\in\B^k$ such that $xA_S=0$. Since $A$
has full row rank it is also the case that $xA\neq 0$. We can thus
conclude that $xA$ is supported on $[k+b]\setminus S$ as required.
\end{proof}

\begin{corollary}\label{cor:woz is mds}
  For any $S\subset[k+b]$ of size $|S|\leq (1-\epsilon)b$, the
  number of matrices $A$ in the $(k,k+b)$-Wozencraft ensemble that
  $A_{[k+b]\setminus S}$ does not have full row rank is smaller than $2^{k-\epsilon
  b}$.
\end{corollary}

\begin{proof}
  Let $y$ be the characteristic vector of $S$.
  In particular, the wight of $y$ is $\leq (1-\epsilon)b$.
  By Corollary~\ref{cor:wozen-union bound}, the number of matrices
  that contain a vector $0\neq y'\leq y$ in the span of their rows
  is at most $(2^{(1-\epsilon)b}-1) \cdot 2^{k-b}< 2^{k-\epsilon b}$.
  By Lemma~\ref{lem:good subset} we see that any other matrix in
  the ensemble has full row rank when we restrict to the columns
  in $[k+b]\setminus S$.
\end{proof}

%

\subsection{The construction}

Let $c,\epsilon>0$ and $0<p<1$ be real numbers. Let $n$ be such
that $$\log n  <  n^{c\epsilon /4} \;\; \text{ and } \;\;
8/\epsilon < n^{(p+\epsilon/2)c\epsilon}.$$ Notice that
$n=(1/c\epsilon)^{O(1/pc\epsilon)}$ satisfies this condition. Let
$k=(1-p-\epsilon/2)\cdot c\log n$, $b=(p+\epsilon/2)\cdot c\log n$
and  $$I=k\cdot \frac{n}{(c\log n)2^{\epsilon b}} =
(1-p-\epsilon/2)\frac{n}{2^{\epsilon b}}.$$ To simplify notation
assume that $k$, $b$ and $I$
are integers.

Our encoding scheme will yield a WOM code of length $n+I$, which,
by the choice of $n$, is at most $n+I < (1+\epsilon/8)n$, and rate
larger than $H(p) + (1-p) - \epsilon$.

\sloppy \paragraph{Step I.} A message in the first round consists
of $n/(c\log n)$ subsets $S_1,\ldots,S_{n/(c\log n)}\subset[c\log
n]$ of size at most $p\cdot(c\log n)$ each. We encode each $S_i$
using its characteristic vector $w_i$ and denote $w=w_1\circ
w_2\circ \ldots\circ w_{n/(c\log n)}\circ \vec{0}_I$, where
$\vec{0}_I$ is the zero vector of length $I$. Reading the message
$S_1,\ldots,S_{n/(c\log n)}$ from $w$ is trivial.

\paragraph{Step II.} Let $x=x_1\circ x_2\circ \ldots\circ
x_{n/(c\log n)}$ be a concatenation of $n/(c\log n)$ vectors of
length $k=(1-p-\epsilon/2) c\log n$ each. Assume that in the first
step we transmitted a word $w$ corresponding to the message
$(S_1,\ldots,S_{n/(c\log n)})$ and that we wish to encode the
message $x$ in the second step.
For each $1\leq i\leq \frac{n}{(c\log n)2^{\epsilon b}}$ we do the
following.

\subparagraph{Step II.i.} Find a matrix $A_{\alpha}$ in the
$(k,k+b)$-Wozencraft ensemble such that for each $(i-1)2^{\epsilon
b}+1\leq j \leq i2^{\epsilon b}$ the submatrix
$(A_{\alpha})_{[c\log n]\setminus S_j}$ has full row rank. Note
that Corollary~\ref{cor:woz is mds} guarantees that such a matrix
exists. Denote this required matrix by $A_{\alpha_i}$.

\subparagraph{Step II.ii.} For $(i-1)2^{\epsilon b}+1\leq j \leq
i2^{\epsilon b}$ find a vector $y_j \in \B^{k+b}=\B^{c\log n}$
such that $A_{\alpha_i} y_j = x_j$ and $y_j|_{S_j} = w_j$. Such a
vector exists by the choice of $A_{\alpha_i}$. The encoding of $x$
is the vector $y_1\circ y_2 \circ \ldots \circ y_{n/(c\log n)}
\circ
\sigma_k(\alpha_1)\circ\ldots\circ\sigma_k(\alpha_{\frac{n}{(c\log
n)2^{\epsilon b}}})$. Observe that the length of the encoding is
$c\log(n) \cdot n/(c\log(n)) +k\cdot\frac{n}{(c\log n)2^{\epsilon
b}} =n+I$. Notice that given such an encoding we can recover $x$
in the following way. Given $(i-1)2^{\epsilon b}+1\leq j \leq
i2^{\epsilon b}$ set $x_j =A_{\alpha_i}y_j$, where $\alpha_i$ is
trivially read from the last $I$ bits of the encoding.

\subsection{Analysis}

\paragraph{Rate.}
From Stirling's formula it follows that the number of messages
transmitted in {\bf Step I.} is at least $(2^{H(p)c\log n -
\log\log n})^{n/(c\log n)} = 2^{H(p)n - n\log\log n/(c\log n)}$.
In {\bf Step II.} it is clear that we encode all messages of
length $kn/(c\log n) = (1-p-\epsilon/2)n$. Thus, the total rate is
\begin{align*}
&((H(p) - \log\log n/(c\log n)) +(1-p - \epsilon/2))n/(n+I)\\
> &((H(p) - \log\log n/(c\log n)) +(1-p -
\epsilon/2))(1-\epsilon/8)\\
> & (H(p) + 1-p) - \epsilon \log_2(3)/8  - \epsilon/2 -\log\log n/(c\log
n)\\
>& H(p)+1-p-\epsilon,
\end{align*}
where in the second inequality we used the fact that
$\max_p(H(p)+1-p) = \log_2 3$. The last inequality follows since
$\log n < n^{c\epsilon/4}$.

\paragraph{Complexity.}
The encoding and decoding in the first step are clearly done in
polynomial time.\footnote{We do not explain how to encode sets as
binary vectors but this is quite easy and clear.}

In the second step, we have to find a ``good'' matrix
$A_{\alpha_i}$ for all sets $S_j$ such that $(i-1)2^{\epsilon
b}+1\leq j \leq i2^{\epsilon b}$. As there are $2^{c\log n} = n^c$
matrices and each has size $k\times c\log n$, we can easily
compute for each of them whether it has full row rank for the set
of columns $[c\log n]\setminus S_j$. Thus, given $i$, we can find
$A_{\alpha_i}$ in time at most $2^{\epsilon b}\cdot n^{c}\cdot
\poly(c\log n)$. Thus, finding all $A_{\alpha_i}$ takes at most
$$\frac{n}{(c\log n)2^{\epsilon b}} \cdot (2^{\epsilon b}\cdot n^{c}\cdot \poly(c\log n))= n^{c+1}\cdot \poly(c\log n).$$
Given $A_{\alpha_i}$ and $w_j$, finding $y_j$ amounts to solving a
system of $k$ linear equations in (at most) $c\log n$ variables
which can be done in time $\poly(c\log n)$. It is also clear that
computing $\sigma_k(\alpha_i)$ requires $\poly(c\log n)$ time.
Thus, the overall complexity is $ n^{c+1}\cdot \poly(c\log n)$.
Decoding is performed by multiplying each of the $A_{\alpha_i}$ by
$2^{\epsilon b}$ vectors so the decoding complexity is at most
$\frac{n}{(c\log n)2^{\epsilon b}}\cdot 2^{\epsilon b} \cdot
\poly(c\log n)=
n\cdot \poly(c\log n)$.\\

Theorem~\ref{thm:main} is an immediate corollary of the above
construction and analysis.

\section{Connection to extractors for bit-fixing sources}\label{sec:extractors}

Currently, our construction is not very practical because of the
large encoding length required to approach capacity. It is an
interesting question to come with ``sensible'' capacity achieving
codes. One approach would be to find, for each $n$, a set of
$\poly(n)$ matrices $\{A_i\}$ of dimensions $(1-p-\epsilon)n\times
n$ such that for each set $S \subset[n]$ of size $|S|=(1-p)n$
there is at least one $A_i$ such that $A_i|_S$ has full row rank.
Using our ideas one immediately gets a code that is (roughly)
$\epsilon$-close to capacity.

One way to try and achieve this goal may be to improve known
constructions of {\em seeded linear extractors for bit-fixing
sources}. An $(n,k)$ bit-fixing source is a uniform distribution
on all strings of the form $\{v \in \B^n \mid v_S = \vec{a}\}$ for
some $S\subset [n]$ of size $n-k$ and $\vec{a}\in\B^{n-k}$. We call such
a source $(S,\vec{a})$-source.

Roughly, a
seeded linear extractor for $(n,k)$ sources that extracts $k-o(k)$
of the entropy, with a seed length $d$, can be viewed as a set of
$2^d$ matrices of dimension $(k-o(k)) \times n$ such that for each
$S\subset [n]$ of size $|S|=n-k$, a $1-\epsilon$ fraction of the matrices $A_i$
satisfy $A_i|_{[n]\setminus S}$ has full row rank.\footnote{Here
we use the assumed linearity of the extractor.}

\begin{definition}\label{def: extractors}
A function $E:\{0,1\}^n\times\{0,1\}^d \to \{0,1\}^m$
is said to be a strong linear seeded $(k,\epsilon)$-extractor for bit fixing sources if the
following properties holds.\footnote{We do not give the most general definition, but
rather a definition that is enough for our needs. For a more general definition see \cite{Rao-thesis}.}
\begin{itemize}
\item For every $r\in \{0,1\}^d$, $E(\cdot,r):\{0,1\}^n \to \{0,1\}^m$
 is a linear function.
\item For every $(n,k)$-source $X$, the distribution $E(X,r)$ is equal to the
uniform distribution on $\{0,1\}^m$ for $(1-\epsilon)$ of the seeds $r$.
\end{itemize}
\end{definition}


Roughly, a seeded linear extractor for $(n,k)$ sources that extracts $k-o(k)$
of the entropy, with a seed length $d$, can be viewed as a set of
$2^d$ matrices of dimension $(k-o(k)) \times n$ such that for each
$S\subset [n]$ of size $|S|=n-k$, $1-\epsilon$ of the matrices $A_i$
satisfy $A_i|_{[n]\setminus S}$ has full row rank.\footnote{Here
we use the assumed linearity of the extractor.}
Note that this is a
stronger requirement than what we need, as we would be fine also
if there was one $A_i$ with this property. Currently, the best
construction of seeded linear extractors for $(n,k)$-bit fixing
sources is given in \cite{RazRV02}, following \cite{Trevisan01},
and has a seed length $d=O(\log^3 n)$. We also refer the reader to
\cite{Rao09} where linear seeded extractors for affine sources are
discussed.

\begin{theorem}[\cite{RazRV02}]\label{thm:RRV}
For every $n,k\in \mathbb{N}$ and $\epsilon>0$,
there is an explicit strong seeded
$(k,\epsilon)$-extractor $\text{Ext} : \{0, 1\}^n \times \{0, 1\}^d
\to \{0,1\}^{k - O(\log^3(n/\epsilon))}$, with $d=O(\log^3(n/\epsilon))$.
\end{theorem}

In the next section we show how one can use the result of
\cite{RazRV02} in order to design encoding schemes for
defective memory.

Going back to our problem, we note that if one could get an
extractor for bit-fixing sources with seed length $d=O(\log n)$
then this will give the required $\poly(n)$ matrices and
potentially yield a ``reasonable'' construction of a capacity
achieving two-write WOM code.

Another relaxation of extractors for bit-fixing sources is to
construct a set of matrices of dimension $(1-p-\epsilon)n\times
n$,  $\cal A$, such that $|{\cal A}|$ can be as large as $|{\cal
A}|=\exp(o(n))$, and that satisfy that given an
$(S,\alpha)$-source we can efficiently find a matrix  $A\in{\cal
A}$ such that $A|_{[n]\setminus S}$ has full row rank. It is not
hard to see that such a set also gives rise to a capacity
achieving WOM codes using a construction similar to ours. Possibly, such
$\cal A$ could be constructed to give more effective WOM codes. In
fact, it may even be the case that one could ``massage'' existing
constructions of seeded extractors for bit-fixing sources so that
given an $(S,\alpha)$-source a ``good'' seed can be efficiently
found.

\section{Memory with defects}\label{sec:defective}

In this section we demonstrate how the ideas raise in Section~\ref{sec:extractors}
can be used to handle defective memory.

A memory containing $n$ cells is said to have $pn$ defects if
$pn$ of the memory cells have some value stored on them that cannot
be changed. We will assume that the person storing data in the
memory is aware of the defects, yet the person reading the memory
cannot distinguish a defective cell from a proper cell.

The main question concerning defective memory is to find a scheme
for storing as much information as possible that can be retrieved efficiently,
no matter where the $pn$ defects are.

\remove{
If defects are random we can use the same idea (treat the defects
as the message from the first round)
}

We will demonstrate a method for dealing with defects that is based on
linear extractors for bit fixing sources. To make the scheme work we will
need to make an additional assumption:
\begin{center} {\bf Our assumption:}
We shall assume that the memory contains
$O(\log^3 n)$ cells that are undamaged and whose identity is known to both
the writer and the reader.
\end{center}
We think that our assumption, although not standard is very reasonable.
For example, we can think of having a very small
and expensive chunk of memory that is highly reliable and a larger memory
that is not as reliable.

\paragraph{The encoding scheme}

\remove{
The first ingredient that we will need is a linear extractor for
bit fixing source. A function $E:\{0,1\}^n\times\{0,1\}^d \to \{0,1\}^m$
is said to be a strong linear seeded $(k,\epsilon)$-extractor for bit fixing sources if the
following properties holds.\footnote{We do not give the most general definition, but
rather a definition that is enough for our needs. For a more general definition see \cite{Rao-thesis}.}
\begin{itemize}
\item For every $r\in \{0,1\}^d$, $E(\cdot,r):\{0,1\}^n \to \{0,1\}^m$
 is a linear function.
\item For every $(n,k)$-source $X$, the distribution $E(X,r)$ is equal to the
uniform distribution on $\{0,1\}^m$ for $(1-\epsilon)$ of the seeds $r$.
\end{itemize}

The following theorem by Raz et al., following \cite{Trevisan01}, guarantees the
existence of strong linear seeded extractors \cite{RazRV02}.
\begin{theorem}[\cite{RazRV02}]\label{thm:RRV}
For every $n,k\in \mathbb{N}$ and $\epsilon>0$,
there is an explicit strong seeded
$(k,\epsilon)$-extractor $\text{Ext} : \{0, 1\}^n \times \{0, 1\}^d
\to \{0,1\}^{k - O(\log^3(n/\epsilon))}$, with $d=O(\log^3(n/\epsilon))$.
\end{theorem}
}

Our scheme will be randomized in nature. The idea is that each
memory with $k=pn$ defects naturally defines an $(n,k)$-source, X,
that is determined by the values in the defective cells. Consider
the extractor $\text{Ext}$ guaranteed by Theorem~\ref{thm:RRV}. We
have that for $(1-\epsilon)$ fraction of the seeds $r$, the linear
map $\text{Ext}:X\to \{0,1\}^{k - O(\log^3(n/\epsilon))}$ has full
rank. (as it induces the uniform distribution on $\{0,1\}^{k -
O(\log^3(n/\epsilon))}$.) In particular, given a string $y\in
\{0,1\}^{(1-p)n- O(\log^3(n/\epsilon))}$, if we pick a seed $r\in
O(\log^3(n/\epsilon))$ at random, then with probability at least
$(1-\epsilon)$ there will be an $x\in X$ such that
$\text{Ext}(x,r)=y$.

Thus, our {\em randomized} encoding scheme will work as follows.
Given the defects, we define the source $X$ (which is simply the
affine space of all $n$-bit strings that have the same value in
the relevant coordinates as the defective memory cells). Given a
string $y\in \{0,1\}^{(1-p)n- O(\log^3(n/\epsilon))}$ that we wish
to store to the memory, we will pick at random $r\in \{0, 1\}^d $,
for $d=O(\log^3(n/\epsilon))$, and check whether $\text{Ext}:X\to
\{0,1\}^{k - O(\log^3(n/\epsilon))}$ has full rank. This will be
the case with probability at least $1-\epsilon$. Once we have
found such $r$, we find $x\in X$ with  $\text{Ext}(x,r)=y$. As
$x\in X$ and $X$ is ``consistent'' with the pattern of defects, we
can write $x$ to the memory. Finally, we write $r$ in the
``clean'' $O(\log^3(n/\epsilon))$ memory cells that we assumed to
have.

The reader in turn, will read the memory $x$ and then $r$ and will recover
$y$ by simply computing $\text{Ext}(x,r)$.

In conclusion, for any constant\footnote{The scheme can in fact work also when
$p=1-o(1)$, and this can be easily deduced from the above, but we present here
the case of $p<1$.} $p<1$ the encoding scheme described above needs
$O(\log^3 n)$ clean memory cells, and then it can store as much
as $(1-p-\delta)n$ bits for any constant $\delta>0$.\footnote{Again, we can take $\delta=o(1)$
but we leave this to the interested reader.}

We summarize this result in the following theorem.

\begin{theorem}\label{thm:defective}
For any constant $p<1$ there is a randomized encoding scheme
that given access to a defective memory of length $n$
containing $pn$ defective cells,
uses $O(\log^3 n)$ clean memory cells,
and can store $(1-p-\delta)n$ bits for any constant $\delta>0$.

The encoding and decoding times for the scheme are polynomial in $n$
and $1/\delta$.
\end{theorem}

\section{Approaching capacity without lookup tables}\label{sec:no lookup}

In this section we describe how one can use the techniques of
\cite{CohenGodlewskiMekx,Wu10,YKSVW10} in order to achieve
codes that approach capacity without paying the cost of
storing huge lookup tables. The reader is referred to Section~\ref{sec:Yaakobi}
for a summary of the basic approach. We will give a self contained treatment
here.

Let $0<p<1$ and $\epsilon$ be real numbers.
Let $A$ be a $(1-p)m\times m$ matrix that has the following property
\begin{center}
{\flushleft{\bf Main property of $A$:}\\} For $(1-\epsilon)$
fraction of the subsets $S\subset [m]$ of size $pm$ it holds that
$A|_{[m]\setminus S}$ has full rank.
\end{center}

Recall that this is exactly the property that is required by
\cite{CohenGodlewskiMekx,Wu10,YKSVW10}. However, while in those
works a lookup table was needed we will show how to trade space for computation
and in particular, our encoding scheme will only need to store the
matrix $A$ itself (whose size is logarithmic in the size of the lookup table).

\paragraph{The encoding scheme}

Let $\Sigma = {[m] \choose pm}$. In words, $\Sigma$
is the collection of all subsets of $[m]$ of size $pm$.
We denote $\sigma = |\Sigma| = {m\choose pm}$.
Let $N=\sigma \cdot m$.
We will construct an encoding scheme for $N$ memory cells.

We denote with $\Sigma_g \subset \Sigma$ (g stands for ``good'') the subset
of $\Sigma$ containing all those sets $S$ for which
$A|_{[m]\setminus S}$ has full rank. We also denote
$\sigma_g = |\Sigma_g| \geq (1-\epsilon)\sigma$.

We let $V = \{0,1\}^{(1-p)m} \setminus \{A \cdot \vec{1}\}$ be the
set of vectors of length $(1-p)m$ that contains all vectors except
the vector $A \cdot \vec{1}$. Clearly $|V|=2^{(1-p)m}-1$.

\subparagraph{The first round:} A message will be an
equidistributed\footnote{From here on we use the term
`equidistributed' to denote words that contain each symbol of the
alphabet the same number of times.} word in $\Sigma^\sigma$.
Namely, it will consist of all $\sigma$ subsets of $[m]$ of size
$pm$ each, such that each subset appears exactly once. We denote
this word as $w=w_1\circ w_2\circ\ldots\circ w_\sigma$ where
$w_i\in \Sigma$. (alternatively, a word is a permutation of
$[\sigma]$.)

To write $w$ to the memory we will view the $N$ cells
as a collection of $\sigma$ groups of $m$ cells each.
We will write the characteristic vector of $w_i$ to the
$m$ bits of $i$th group.

\subparagraph{The second round:}

A message in the second round consists of $\sigma_g$ vectors
from $V$. That is, $x=x_1\circ\ldots\circ x_{\sigma_g}$, where
$x_i\in V$.

To write $x$ to memory we first go over all the memory cells
and check which coordinates belong
to $\Sigma_g$. According to our scheme there are exactly $\sigma_g$ such
$m$-tuples. Consider the $i$th $m$-tuple that belongs to $\Sigma_g$.
Assume that it encodes the subset $S \subset [m]$ (recall that $|S|=pm$).
Let $w_S$ be its characteristic vector. (note that this $m$-tuple
stores $w_S$.)
We will find the unique $y\in \{0,1\}^m\setminus \vec{1}$ such that
$Ay=x_i$ and $y|_S=w_S$. Such a $y$ exists since $A|_{[m]\setminus S}$
has full rank.

After writing $x$ to memory in this way, we change the value of the other
$\sigma-\sigma_g$ $m$-tuples to $1111...1$. Namely, whenever an $m$-tuple
stored a set not from $\Sigma_g$ we change its value in the second write
to $\vec{1}$.

Recovering $x$ is quite easy. We ignore all $m$-tuples that
contain the all $1$ vector. We are thus left with $\sigma_g$
$m$-tuples. If $y_i$ is the $m$-bit vector stored at the $i$th
``good'' $m$-tuple then $x_i=Ay_i$.

\subparagraph{Analysis}

The rate of the first round is
$$\frac{\log(\sigma!)}{N} =\frac{\log(\sigma!)}{m\sigma} = \frac{\log(\sigma)}{m} - O(\frac{1}{m})
= \frac{\log{m \choose pm}}{m} - O(\frac{1}{m}) =
 H(p) - O(\frac{1}{m}) .$$

In the second round we get
rate
\begin{eqnarray*}
 \frac{\log((2^{(1-p)m}-1)^{\sigma_g})}{N} =
\frac{\sigma_g \cdot \log(2^{(1-p)m}-1))}{\sigma m} &=&
\frac{(1-\epsilon)\log(2^{(1-p)m}-1))}{m} \\&=& (1-\epsilon)(1-p) - O(\exp(-(1-p)m)).
\end{eqnarray*}

Hence, the
overall rate of our construction is $$H(p)+ (1-\epsilon)(1-p) + O(1/m).$$
Notice that the construction of  \cite{YKSVW10} gives rate
$\log(1-\epsilon) + H(p) + (1-p)$.
Thus, the loss of our construction
is at most
$$ \epsilon p + O(1/m) - \log(1-\epsilon) = O(\epsilon + 1/m).$$
Note, that if \cite{YKSVW10} get $\epsilon$ close to capacity then
we must have $m=\poly(1/\epsilon)$
and so our codes get $O(\epsilon)$ close to capacity.
To see that it must be the case that $m=\poly(1/\epsilon)$
we note that by probabilistic argument it is not hard to show
that, say, $\sigma_g \leq \sigma/2$. Thus, the rate achieved by
\cite{YKSVW10} is at most $H(p) + (1-p) -1/m$, and so to be $\epsilon$-close
to capacity (which is $\max_p(H(p) +(1-p)$), we must have $m\geq 1/\epsilon$.


Concluding, our scheme enables a tradeoff: for the \cite{YKSVW10} scheme
to be $\epsilon$-close to capacity we need $m=\poly(1/\epsilon)$ and therefore
the size of the lookup table that they need to store is $\exp(1/\epsilon)$.
In our scheme, the block length is $\exp(1/\epsilon)$ (compared to
$\poly(1/\epsilon)$ in  \cite{YKSVW10}), but we do not need to store a lookup
table.

\section{3-write binary WOM codes}\label{sec:3-write}

In this section we give an asymptotic construction of a $3$-write
WOM code over the binary alphabet that achieves rate larger than
$1.809-\epsilon$. Currently, the best known methods give rate
$1.61$ \cite{KYSVW10} and provably cannot yield rate better than
$1.661$. The main drawback of our construction is that the block
length has to be very large in order to approach this rate.
Namely, to be $\epsilon$ close to the rate the block length has to
be exponentially large in $1/\epsilon$.

An important ingredient in our construction is
a $2$-write binary WOM code
due to Rivest and Shamir \cite{RivestShamir82} that we recall next. The block length
of the Rivest-Shamir construction is $3$ and the rate is $4/3$.
In each round we write one of four symbols $\{0,1,2,3\}$ which are
encoded as follows.

\begin{table}[h]\begin{center}
\begin{tabular}{|c|c|c|}
\hline
Symbol & weight $0/1$ & weight $2/3$\\ \hline
0 & 000 & 111\\ \hline
1 & 001 & 110\\ \hline
2 & 010 & 101\\ \hline
3 & 100 & 011\\ \hline
\end{tabular}
\caption{The Rivest-Shamir encoding}
\end{center}
\end{table}

In the first round we write for each symbol the value in the `weight $0/1$' column.
In the second round we use for each symbol, the minimal possible weight
representing it and that is a `legal' write. For example, if in the first round
the symbol was $2$ and at the second round it was $1$ then we first write $010$
and then $110$. On the other hand, if in the first round the symbol was
$0$ and in the second round it was $1$ then we first write $000$ and then $001$.

\paragraph{The basic idea.}

We now describe our approach for constructing a $3$-write WOM
code. Let $n$ and $m$ be integers such that $n=12m$. We shall construct a code
with block length $n$. We first partition the $n$ cells to
$4m$ groups of $3$ cells each. A message in the first round
corresponds to a word $w_1 \in \{0,1,2,3\}^{4m}$ such that each
symbol appears in $w_1$ exactly $m$ times. (we will later ``play''
with this distribution.) We encode $w_1$ using the Rivest-Shamir
scheme, where we use the $i$th triplet to encode $(w_1)_i$.
The second round is the same as the first round. I.e. we get
$w_2 \in \{0,1,2,3\}^{4m}$ that is equidistributed and write it
using the Rivest-Shamir scheme.

Before we describe the third round let us calculate an upper
bound on the number of memory cells that have value $1$, i.e., those
cells that we cannot use in the third write.

Notice that according to the Rivest-Shamir encoding scheme, a triplet of
cells (among the $4m$ triplets) stores $111$ if and only if, in
the first round it stored a symbol from $\{1,2,3\}$ and in the
second round it stored a zero. Similarly, a triplet has weight $2$
only if in both rounds it stored a symbol from $\{1,2,3\}$. We
also note, that a triplet that stored zero in the first round,
will store a word of weight at most one after the second write.
Since in the second round we had only $m$ zeros and in the first
round we wrote only $3m$ values different than zero, the weight of
the stored word is at most
$$m\times 3 + (3m-m)\times 2 + m\times 1 = 8m = 2n/3 .$$
Thus, we still have $n/3$ zeros that we can potentially use in the third write.
We can now use the same idea as in the construction of capacity achieving
$2$-write WOM codes and with the help of the Wozencraft ensemble achieve
rate $(1/3 - o(1))$ for the third write.\footnote{This step actually
involves concatenating many copies of the construction with itself
to achieve reasonable running time, and as a result the block length blows
to $\exp(1/\epsilon)$.} Thus, the overall rate of this construction
is $2/3 + 2/3 + 1/3 -o(1) = 5/3 - o(1)$. As before, in order to
be $\epsilon$-close to $5/3$ we need to take $n=\exp(1/\epsilon)$.
Note that this idea already yields codes that beat the
best possible rate one can hope to achieve using the methods of
Kayser et al. \cite{KYSVW10}.

\paragraph{Improvement I.}

One improvement can be achieved by modifying the distribution of
symbols in the messages of the first round. Specifically, let us
only consider messages $w_1\in \{0,1,2,3\}^{4m}$ that have at
least $4pm$ zeros (for some parameter $p$). The rate of the first
round is thus $(1/3)(H(p)+(1-p)\log(3))$. In the second round we
again write an equidistributed word $w_2$. Calculating, we get
that the number of nonzero memory cells after the second write is
at most
$$m\times 3 + (4(1-p)m - m)\times 2 + 4pm\times 1 = 9m -4pm\;.$$
Thus, in the third round we can achieve rate
$\frac{3m+4pm}{12m} - o(1) = p/3 + 1/4 -o(1)$. Hence, the overall rate is
$$(1/3)\cdot (H(p) + (1-p)\log(3)) + (2/3) + (p/3 + 1/4) -o(1)\;.$$
Maximizing over $p$ we get rate larger than $1.69$ when $p=2/5$.

\remove{
Note that we can improve it if we put a less zeros in the second round
i.e., a message will have $p=1/7$ many zeroes and $2/7$ of any other number.
The rate of the second round will be $(H(p) + (1-p)\log(3))4m$. The weight
of the word after the 2nd write will then be $3*4pm + (3m-4pm)*2 + m = 7m + 4pm$.
Thus, in the third round we will write a word of weight $5m-4pm$.
The rate is $2/3 + (1/3)(H(p) + (1-p)\log(3)) + (5-4p)/12$.
Maximizing gives $p=1/7$ and overall rate is roughly 1.68.
}

\paragraph{Improvement II.}

Note that so far we
always assumed that the worst had happened, i.e., that all
the zero symbols of $w_2$ were assigned to cells that stored
a value among $\{1,2,3\}$. We now show how one can assume that
the ``average'' case has happened using the aid of two additional memory cells.

Let $n=12m$ and $N=n+2$. As before, let $p$ be a parameter to be
determined later. A message in the first round is some $w_1\in
\{0,1,2,3\}^{4m}$ that has at least $4pm$ zeros. Again, we use the
Rivest-Shamir encoding to store $w_1$ on the first $n$ memory
cells. We define the set $I = \{i\mid (w_1)_i\neq 0\}$. Notice
that $|I|\leq 4(1-p)m$. In the second round we get a word $w_2\in
\{0,1,2,3\}^{4m}$ which is equidistributed. We identify an element
$\alpha\in\{0,1,2,3\}$ that appears the least number of times in
$(w_2)|_I$. I.e., it is the symbol that is repeated the least
number of times in $w_2$ when we only consider those coordinates
in $I$. We would like this $\alpha$ to be $0$ but this is not necessarily
the case. So, to overcome this we change the meaning of the symbols of
$w_2$ in the following way:
We write
$\alpha$ in the last two memory cells (say, using its binary
representation) and define a new word $w'_2 \in \{0,1,2,3\}^{4m}$
from $w_2$ by replacing each appearance of zero with $\alpha$ and
vice versa. We now use the Rivest-Shamir encoding scheme to store
$w'_2$. It is clear that we can recover $w'_2$ and $\alpha$ from the
stored information and therefore we can also recover $w_2$ (by replacing $0$ and $\alpha$).
The advantage of this trick is that the weight of the stored word is at most
$$\frac{1}{4}\cdot4(1-p)m\times 3 + \frac{3}{4}\cdot4(1-p)m\times 2 +
\frac{1}{4}\cdot 4pm\times 0 + \frac{3}{4}\cdot 4pm\times 1 =
(9-6p)m = (3/4 -p/2)n\;.$$ Indeed, in $w'_2$
the value zero appears in at most $|I|/4$ of the cells in $I$.
Thus, at most $\frac{1}{4}\cdot 4(1-p)m$ triplets will have the value $111$.
Moreover, the rest of the zeros (remember that $w'_2$ had exactly
$m$ zeros) will have to be stored in triplets that already contain
the zero triplet so they will leave those cells unchanged (and of weight zero). As a
result, in the third round we will be able to store $(1/4 +p/2)n
-o(n)$ bits (this is the number of untouched memory cells after the second
round). To summarize, the rate that we get
is\footnote{ The additional two coordinates have no affect on the
asymptotic rate.}
$$ (1/3)\cdot (H(p) + (1-p)\log(3)) + (2/3) + (1/4 +p/2) -o(1)\;.$$
Maximizing over $p$ we get that for $p\approx 0.485$ the rate is larger than $1.76$.

\paragraph{Improvement III.}
The last improvement comes from noticing that so far we assumed
that all the triplets that had weight $1$ after the first write,
have weight at least $2$ after the second write. This can be taken
care of by further permuting some of the values of $w_2$.
Towards this goal we shall make use of the following notation.
For a word $w\in \{0,1,2,3\}^{4m}$ let
$$I_0(w) = \{ i \mid (w_1)_i\neq 0 \;\text{ and }\; w_i=0\} $$
and
$$I_{=}(w) = \{ i \mid (w_1)_i\neq 0 \;\text{ and }\; w_i=(w_1)_i\}\;.$$
For a permutation $\pi:\{0,1,2,3\}\to \{0,1,2,3\}$ define the word $w_\pi$
to be
$(w_\pi)_i = \pi((w)_i)$.

Let $n=12m$ and $N=n+5$. As before, let $p$ be a parameter to be determined later.
A message in the first round is some $w_1\in \{0,1,2,3\}^{4m}$ that
has at least $4pm$ zeros.
We use the Rivest-Shamir encoding scheme to store $w_1$ on the first $n$ memory cells.
A message for the second write is  $w_2\in \{0,1,2,3\}^{4m}$.
 We now
look for a permutations $\pi:\{0,1,2,3\}\to \{0,1,2,3\}$ such that
$|I_0(w_\pi)|\leq \frac{1}{4}\cdot4(1-p)m=(1-p)m$ and
$|I_{=}(w_\pi)| \geq \frac{1}{4}\cdot 4(1-p)m=(1-p)m$. Observe
that such a $\pi$ always exists. Indeed, as before we can first
find $\pi^{-1}(0)$ by looking for the value that appears the least
number of times in $w_2$ on the coordinates where $w_1$ is not
zero. Let us denote this value with $\alpha$. We now consider only
permutations that send $\alpha$ to $0$. After we apply this
transformation to $w_2$ (namely, switch between $\alpha$ and $0$)
we denote the resulting word by $w'_2$. Let $J=\{i\mid (w_1)_i\neq
0  \;\text{ and }\; (w'_2)_i\neq 0\}$. I.e., $J$ is the set of
coordinates that we need to consider in order to satisfy
$|I_{=}(w_\pi)| \geq (1-p)m$. By the choice of $\alpha$ we get
that $|J|\geq 4(1-p)m - \frac{1}{4}\cdot 4(1-p)m = 3(1-p)m$. Now,
among all permutations that send $\alpha$ to zero, let us pick one
at random and compute the expected size  $|I_{=}((w'_2)_\pi)|$.
Notice, that when picking a permutation at random the probability
that a coordinate $i\in J$, will satisfy $(w_1)_i=((w'_2)_\pi)_i$
is exactly $1/3$. Thus, the expected number of coordinates in $J$
that fall into $I_{=}((w'_2)_\pi)$ is $|J|/3$. In particular there
exists a permutation $\pi$ that achieves $|I_{=}((w'_2)_\pi)|\geq
|J|/3 \geq 3(1-p)m/3 = (1-p)m$. Let $\pi_0$ be this permutation.
We use the last $5$ memory cells to encode $\pi_0$. As there are
$4!=24$ permutations, this can be easily done.

Now, we consider the word $(w'_2)_{\pi_0}$ and write it to the first $n$ memory cells
using the Rivest-Shamir scheme. Notice that after this second write, the
weight of the word stored in the first $n$ memory cells is at most
\begin{eqnarray*}
\frac{1}{4}\cdot 4(1-p)m\times 3 + \frac{2}{3}\cdot 3(1-p)m\times 2 +
\frac{1}{3}\cdot 3(1-p)m\times 1 + \frac{1}{4}\cdot 4pm\times 0 + \frac{3}{4}\cdot 4pm\times 1\\
= (8-5p)m = (8-5p)n/12 \;,
\end{eqnarray*}
where the term $\frac{1}{3}\cdot 3(1-p)m\times 1$ comes from the
contribution of the coordinates in $I_{=}((w'_2)_{\pi_0})$.
Thus, in the third write we can store $(4+5p)n/12 -o(n)$
bits. The total rate is thus
$$(1/3)\cdot (H(p) + (1-p)\log(3)) + (2/3) + (4+5p)/12 -o(1)\;.$$
Maximizing, we get that for $p\approx 0.442$ the rate is
larger than $1.809$.\\

The proof of Theorem~\ref{thm:3-write}
easily follows from the construction above.

\subsection{Discussion}

The construction above yields $3$-write WOM codes that have rate
that is $\epsilon$ close to $1.809$ for block length
roughly $\exp(1/\epsilon)$. In Theorem~\ref{thm:main}
we showed how one can achieve capacity for the case of
$2$-write WOM codes with such a block length.
In contrast, for $3$-write WOM codes over the binary
alphabet the capacity is $\log(4)=2$. Thus, even with
a block length of $\exp(1/\epsilon)$ we fail to reach capacity.
As described in Section~\ref{sec:discussion}
we can achieve capacity by letting the block length
grow like $\exp(\exp(1/\epsilon))$. It is an interesting
question to achieve capacity for $3$-write WOM codes
with a shorter block length.\\

An important ingredient in our construction is the Rivest-Shamir
encoding scheme. Although this scheme does not give the best
$2$-write WOM code we used it as it is easy to analyze and understand the
weight of the stored word after the second write. It may be
possible to obtain
improved asymptotic results (and perhaps even more explicit
constructions)
 by studying existing schemes of $2$-write WOM codes
that beat the Rivest-Shamir construction.

\section*{Acknowledgements}

We are grateful to Eitan Yaakobi for many helpful discussions and
for carefully reading and commenting on an earlier version of this
paper.
We also thank Eitan for several pointers to the literature. We
thank  Alexander Barg, Madhu Sudan and Gilles Z\'{e}mor for
helpful discussions on WOM codes. This work was partially done
while the author was visiting the Bernoulli center at EPFL. We
thank the Bernoulli center for its hospitality.

\bibliographystyle{alpha}
\newcommand{\etalchar}[1]{$^{#1}$}

\end{document}